\documentclass[journal]{IEEEtran}
\usepackage{amsmath,amsfonts,amssymb,euscript, graphicx, 
epsfig,
enumerate,float,afterpage, subfigure, ifthen}%

\newtheorem{thm}{Theorem}
\newtheorem{cor}{Corollary}
\newtheorem{lem}{Lemma}

\newtheorem{fact}{Fact}

\newcommand{\expect}[1]{\mathbb{E}\left\{#1\right\}}
\newcommand{\defequiv}{\mbox{\raisebox{-.3ex}{$\overset{\vartriangle}{=}$}}}

\newcommand{\bv}[1]{{\boldsymbol{#1} }}
\newcommand{\script}[1]{{{\cal{#1} }}}

\newcommand{\WCG}{\script{WC}(\script{G})}

\begin{document}

\title
  {Dynamic Index Coding for Wireless Broadcast Networks}
\author{Michael J. Neely , Arash Saber Tehrani , Zhen Zhang$\vspace{-.3in}$
\thanks{The authors are with the  Electrical Engineering department at the University
of Southern California, Los Angeles, CA.} 
\thanks{This material is supported in part  by one or more of 
the following: the DARPA IT-MANET program
grant W911NF-07-0028, the NSF Career grant CCF-0747525,  NSF grant 0964479, the 
Network Science Collaborative Technology Alliance sponsored
by the U.S. Army Research Laboratory W911NF-09-2-0053.}}

\markboth{}{Neely}

\maketitle

\begin{abstract}   
We consider a wireless broadcast station that transmits packets
to multiple users.  The packet requests for each user may overlap, 
and some users may already have certain packets.  This presents
a problem of broadcasting in the presence of side information, 
and is a generalization of the well known (and unsolved) index coding problem
of information theory.    
Rather than achieving the full capacity region, we develop a \emph{code-constrained
capacity region},  which restricts attention to a pre-specified set of 
coding actions.  We develop a dynamic max-weight algorithm that 
allows for random packet arrivals and supports any traffic inside the
code-constrained capacity region.  Further, we provide a simple
set of codes based on cycles in the underlying demand graph. We
show these codes are optimal for a class of broadcast relay
problems.
\end{abstract}

\section{Introduction} 

Consider a wireless broadcast station that transmits packets to $N$ wireless users.  
Packets randomly arrive to the broadcast station.  Each packet $p$ is desired
by one or more users in the set $\{1, \ldots, N\}$. 
Further, there may be one or more users that already have the packet stored in their cache. 
 The broadcast station must efficiently transmit all packets to their desired users. 
 We assume time is slotted with unit slots $t \in \{0, 1, 2, \ldots\}$, and that a single packet 
 can be transmitted by the broadcast station on every slot.  This packet is received error-free
 at all users.  We assume that only the broadcast station can transmit, so that users cannot transmit 
 to each other. 

If the broadcast station has $P$ packets at time 0, and no more packets arrive, then the mission
can easily be completed in $P$ slots by transmitting the packets one at a time.   However, this 
approach ignores the side-information available at each user.  Indeed, it is often possible to complete
the mission in fewer than $P$ slots if packets are allowed to be  \emph{mixed}  
before transmission.  A simple and well known example for 2 users is the following:  Suppose
user $1$ has packet $B$ but wants
packet $A$, while user $2$ has packet $A$ but wants packet $B$.  Sending each packet individually
would take 2 slots, but these demands can be met in just one slot by transmitting the mixed
packet $A+B$, the bit-wise  XOR of $A$ and $B$.  Such examples are introduced in \cite{two-way-net-code}\cite{cope-05}\cite{cope-06} 
in the 
context of \emph{wireless network coding}. 

The general problem,  where each packet is contained as side information in an arbitrary subset of the $N$ users, 
is much more complex.  This problem is introduced by Birk and Kol in \cite{birk-index98}\cite{birk-index-it06},  and is
known as the \emph{index coding problem}. 
Methods for completing a general index coding mission in minimum time are unknown.  However, 
the recent work \cite{bar-yossef-index-code2011} 
shows that if one restricts to a class of linear codes, then the minimum time
is equal to the minimum rank of a certain matrix completion problem.  The matrix completion problem
is NP-hard in general, and hence index coding is complex even when 
restricted to a simpler class of codes.

Nevertheless, it is important to develop systematic approaches to these problems.  
That is because current wireless cellular systems cannot handle the huge traffic
demands that are expected in the near future.  This is largely due to the 
consistent growth of wireless video traffic.  Fortunately, much of the traffic is for \emph{popular content}.
That is, users often download the same information.  Thus, it is quite likely that a system of $N$ users
will have many instances of side information, where some users already have packets that others want. 
This naturally creates an index coding situation.  Thus, index coding is both rich in its mathematical complexity
and crucial for supporting future wireless traffic. 

The problem we consider in this paper is even more complex because packets can arrive
randomly over time.   This is a practical scenario and creates the need for a \emph{dynamic} approach to 
index coding.  
We assume there are $M$ \emph{traffic types}, where a type is defined by the subset of users
that \emph{desire} the packets and the subset that \emph{already has} the packets. Let $\lambda_m$ be the 
arrival rate, in packets/slot, for type $M$ traffic.   We approach this problem by restricting coding actions to an 
abstract set $\script{A}$.  We then show how to achieve the \emph{code constrained capacity region}
$\Lambda_{\script{A}}$, being the set of all rate vectors $(\lambda_m)_{m=1}^M$ that can be supported using coding
actions in the set $\script{A}$.   The set $\Lambda_{\script{A}}$ is typically a strict subset of the 
\emph{capacity region} $\Lambda$, which does not restrict the 
type of coding action.  Our work can be applied to any set $\script{A}$, and hence can be used in conjunction
with any desired codes.   However, we focus attention on a simple class of codes that involve only bit-wise XOR operations, 
based on cycles in the underlying demand graph.  In special cases of \emph{broadcast relay problems}, we show that these
codes can achieve the full capacity region $\Lambda$.

The capacity region $\Lambda$ is directly related to
the conceptually simpler \emph{static} problem of clearing a fixed batch of packets in minimum time.  Further, index coding
concepts are most easily developed 
in terms of the static problem.  Thus, this paper is divided into two parts:  We first introduce
the index coding problem in the static case, and we describe example coding actions in that case.  
Section \ref{section:dynamic} extends to the dynamic case and develops two max-weight index coding techniques, 
one that requires knowledge of the arrival rates $(\lambda_m)$, and one that does not.  The algorithms here are general 
and can also be used in other types of networks where controllers make sequences of actions, each action taking a different 
number of slots and delivering a different vector of packets. 

While the static index coding problem has been studied before \cite{bar-yossef-index-code2011}\cite{birk-index-it06}\cite{birk-index98}, 
our work provides new insight even in the
static case.  We introduce a new directed bipartite demand graph
that allows for arbitrary demand subsets and possibly ``multiple multicast''
situations, where some packets are desired by more than one user.  We also form a useful \emph{weighted compressed graph} 
that facilitates the solution to the minimum clearance time problem in certain cases. 
This extends the graph models in \cite{bar-yossef-index-code2011}, which 
do  not consider the possibility of multiple 
multicast sessions.  Work in
\cite{bar-yossef-index-code2011} 
develops a maximum acyclic subgraph bound on minimum clearance time for problems without multiple
multicast sessions. We extend this bound to our general problem using a different proof technique.  Further, we 
consider a class of \emph{broadcast relay problems} for which the bound can be achieved with equality.  

The next section introduces index coding in the static case,  shows its relation to a bipartite demand graph, and 
presents the acyclic subgraph bound.  Section \ref{section:dynamic}  introduces the general dynamic formulation and develops
our max-weight algorithms.   Section \ref{section:broadcast} considers an important class of 
broadcast relay networks for which a simple set of codes are optimal.

\section{The Static Minimum Clearance Time Problem}

This section introduces the index coding problem in the static case, where we want to clear a fixed
batch of packets in minimum time.  
Consider a wireless system with $N$ users, $P$ packets, and a single broadcast station. 
We assume $N$ and $P$ are positive integers. Let 
$\script{N}$ and $\script{P}$ represent the set of users and packets, respectively: 
\[ \script{N} = \{1, \ldots, N\} \: \: , \: \: \script{P} = \{1, \ldots, P\} \] 
The broadcast station has all packets in the set $\script{P}$.  Each user $n \in \script{N}$
\emph{has} an arbitrary subset of packets $\script{H}_n \subseteq \script{P}$, and \emph{wants
to receive} an arbitrary subset of packets $\script{R}_n \subseteq \script{P}$, where $\script{H}_n \cap \script{R}_n = \phi$,
where $\phi$ represents the empty set.  Assume that all packets consist of $B$ bits, all packets 
are independent of each other, and the $B$-bit binary string for each packet is uniformly distributed
over each of the $2^B$ possibilities. 

We can represent this system by a \emph{directed bipartite demand graph} $\script{G}$ defined as follows (see Fig. \ref{fig:bipartite-dag}): 
\begin{itemize} 
\item User nodes $\script{N}$ are on the left. 
\item  Packet nodes $\script{P}$ are on the right.
\item A directed 
link $(n,p)$ from a user node  $n \in \script{N}$ to a packet node $p \in\script{P}$ exists if and only if 
user $n$ has packet $p$. That is, 
if and only if $p \in \script{H}_n$. 
\item A directed link $(p,n)$ from a packet node $p \in\script{P}$ to a user node $n \in \script{N}$ exists if and only if user $n$ wants
to receive  packet $p$.  That is, if and only if $p \in \script{R}_n$. 
\end{itemize} 
As an example for the 3-user, 5-packet graph of Fig. \ref{fig:bipartite-dag}, the \emph{have} and \emph{receive} sets for nodes
1 and 2 are: 
\begin{eqnarray*}
\script{H}_1 = \{5\} \: \: , \: \: \script{R}_1 = \{1, 2 \} \\
\script{H}_2 = \phi \: \:  , \: \:  \script{R}_2 = \{1, 2, 4\}  
\end{eqnarray*}

\begin{figure}[t]
   \centering
   \includegraphics[height=1.5in, width=1.5in]{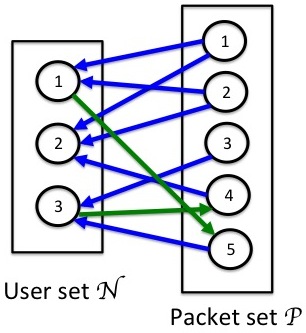} 
   \caption{An example directed bipartite demand graph with $3$ users and $5$ packets.}
   \label{fig:bipartite-dag}
\end{figure}

We restrict attention to packets that at least one node wants. Thus, without loss of
generality, throughout we assume the graph $\script{G}$ is such that all packet nodes $p \in \script{P}$ 
contain at least one outgoing link.   Thus: 
\begin{equation} \label{eq:P} 
 \script{P} = \{1, \ldots, P\} = \cup_{n=1}^N \script{R}_n 
 \end{equation} 

In this static problem, the broadcast station has all packets in the set $\script{P}$ at time $0$, and no more packets ever arrive.  
Every
slot $t \in \{0, 1, 2,\ldots\}$ the broadcast station can transmit one $B$-bit message over the broadcast channel.  
This message is received without
error at all of the user nodes in the set $\script{N}$.   The goal is for the broadcast station to send messages 
until all nodes receive the packets they desire.  

Define a \emph{mission-completing coding action with $T$ slots} to 
be a sequence of messages that the broadcast station transmits
over the course of $T$ slots, such that all users are able to decode their desired packets at the end of the $T$
slots.   We restrict attention to deterministic zero-error codes that enable decoding with probability 1.
The initial information held by each user  $n \in \script{N}$ is given by the set of packets $\script{H}_n$ (possibly empty). 
Let $\script{M} \defequiv \{\script{M}_1, \ldots, \script{M}_T\}$ represent the messages transmitted by the broadcast station over the
course of the $T$ slot coding action.  
At the end of this action, each node $n \in \script{N}$ has information $\{ \script{H}_n, \script{M}\}$. 
Because the coding action is assumed to complete the mission, this information is enough for each node $n$ to decode
its desired packets $\script{R}_n$.  That is, we can write:
\begin{equation} \label{eq:decoding-equivalence}
 \{\script{H}_n, \script{M}\} \iff \{\script{H}_n, \script{M}, \script{R}_n\} 
 \end{equation} 
where the above represents \emph{equivalence in the information set}, meaning that the information on the left-hand-side
can be perfectly reconstructed from the information on the right-hand-side, and vice versa.  Clearly the information on the 
left in \eqref{eq:decoding-equivalence} 
is a subset of the information on the right, and hence can trivially be reconstructed.  The information on the right 
in \eqref{eq:decoding-equivalence} 
can be reconstructed from that on the left because the code is mission-completing. 

For a given graph $\script{G}$ with $P$ packet nodes, define $T_{min}(\script{G})$ as the \emph{minimum clearance time} of the graph, 
being the minimum number of slots required to complete the 
mission, considering all possible coding techniques.  Clearly $T_{min}(\script{G}) \leq P$.  Our goal is to understand
$T_{min}(\script{G})$. 

For a directed graph, we say that a 
\emph{simple directed cycle of length $K$}  is a sequence of nodes $\{n_1, n_2, \ldots, n_K, n_1\}$ 
such that $(n_i, n_{i+1})$ is a link in the graph for all $i \in \{1, \ldots, K-1\}$, $(n_K, n_1)$ is a link in the graph, and 
all nodes $\{n_1, \ldots, n_K\}$ involved in the cycle are distinct.  For simplicity, throughout this paper we use the 
term \emph{cycle} to represent a simple directed cycle. 
We say that the graph $\script{G}$ is \emph{acyclic} if it contains no cycles.   Note that 
directed acyclic graphs have a much different structure than undirected acyclic graphs.  Indeed, the graph in  
Fig. \ref{fig:bipartite-dag} is acyclic even though its undirected counterpart (formed by replacing all directed links
with undirected links) has cycles.  

Our first result is to prove that if a directed bipartite demand graph $\script{G}$ 
is acyclic, then coding cannot reduce the minimum clearance time.   This result was first proven in \cite{bar-yossef-index-code2011} 
in the case without
``multiple multicasts,'' so that each packet is desired by at most one user.  That result uses an argument based on machinery
of the mutual information function.  It also treats a more general case where codes can have errors.  Further, their proof is developed
as a consequence of a more general and more complex result.  Our work restricts to zero-error codes, but allows the possibility
of multiple-multicast sessions.  We also use a different proof technique which emphasizes
the logical consequences of users being able to decode their information.  Our proof uses only
 the following two facts: 

\begin{fact} \label{fact:leaf} Every directed acyclic graph with a finite number of nodes 
has at least one node with no outgoing links.  Such 
a node is called a ``leaf'' node. 
\end{fact} 

\begin{fact} \label{fact:onelink} If the graph  contains only one user node, then $T_{min}(\script{G}) = P$, where $P$ is the number of 
packets that this user desires. 
\end{fact} 

Fact \ref{fact:leaf} follows simply by starting at any node in the graph and traversing a path from node to node, using any
outgoing link, until we find a leaf node (such a path cannot continue forever because the graph is finite and has
no cycles).   Fact \ref{fact:onelink} is a basic information theory observation about the capacity of a single error-free link.

\begin{thm} \label{thm:acyclic}  If the graph $\script{G}$ is acyclic, then $T_{min}(\script{G}) = P$, where $P$ is the total number of 
packets in the graph. 
\end{thm} 
\begin{proof} 
See Appendix A.
\end{proof} 

As an example, because the graph $\script{G}$ in Fig. \ref{fig:bipartite-dag} is acyclic, we have $T_{min}(\script{G}) = 5$. 
Theorem \ref{thm:acyclic} shows that coding cannot help if $\script{G}$ is acyclic, so that the best one can do is just transmit
all packets one at a time.   Therefore, any type of coding must exploit cycles on the demand graph. 

\subsection{Lower Bounds from Acyclic Subgraphs} 
 
 Theorem \ref{thm:acyclic} provides a simple lower bound on $T_{min}(\script{G})$ for any graph $\script{G}$. 
 Consider a graph $\script{G}$, and form a subgraph $\script{G}'$ by performing one or more of the following \emph{pruning
 operations}: 
 \begin{itemize} 
 \item Remove a packet node, and all of its incoming and outgoing links. 
 \item Remove a user node, and all of its incoming and outgoing links. 
 \item Remove a packet-to-user link $(p,n)$. 
 \end{itemize} 
 After performing these operations, we must also delete any residual packets that have no outgoing links. 
 Any sequence of messages that completes the mission for the original graph 
 $\script{G}$ will \emph{also} complete the mission for the subgraph $\script{G}'$. 
 This leads to the following simple lemma. 
 
  \begin{lem}  For any subgraph $\script{G}'$ formed from a graph $\script{G}$ by one or more of
  the above pruning operations, we have: 
  \[ T_{min}(\script{G}') \leq T_{min}(\script{G}) \]
 \end{lem} 
 
 Combining this lemma with Theorem \ref{thm:acyclic}, we see that we can take a general graph $\script{G}$ with 
 cycles, and then perform the above pruning operations to reduce to an acyclic subgraph $\script{G}'$.  Then $T_{min}(\script{G})$ is
 lower bounded by the number of packets in this subgraph.   Thus, the best lower bound corresponds to the acyclic subgraph 
generated from the above operations, and that has 
 the largest number of remaining packets.  Note that the above pruning operations do not include the removal of a user-to-packet 
 link $(n,p)$ (without removing
 either the entire user or the entire packet), because such links represent side information that 
 can be helpful to the mission.

\subsection{Particular code actions} \label{section:cyclic-codes} 

Because the general index coding problem is difficult, 
it is useful to restrict the solution space to consider only 
sequences of simple types of coding actions.  Recall that coding actions must exploit
cycles.  One natural action is the following: 
Suppose we have a cycle in $\script{G}$ that involves a subset of $K$ users.  For simplicity label the users $\{1, \ldots, K\}$.
In the cycle, user 2 wants to receive a packet $X_1$ that user 1 has, user 3 wants to receive a packet $X_2$ that user 2
has, and so on. Finally,  user $1$ wants to receive
a packet $X_K$ that user $K$ has.    
The structure can be represented by: 
\begin{equation} \label{eq:cycle-example} 
 1 \rightarrow 2 \rightarrow 3 \rightarrow \ldots \rightarrow K \rightarrow 1 
 \end{equation} 
where an arrow from one user to another means the left user has a packet the right user wants. 
Of course, the users in this cycle may want many other packets, but we are restricting attention only to the 
packets $X_1, \ldots, X_K$.  Assume these packets are all distinct. 

In such a case, we can  satisfy all $K$ users in the cycle with the following $K-1$ transmissions: 
For each $k \in \{1, \ldots, K-1\}$, the broadcast station transmits a message $\script{M}_k \defequiv X_{k} + X_{k+1}$, 
where addition represents the mod-2 summation of the bits in the packets.   Each user $k \in \{2, \ldots, K\}$ receives its
desired information by adding $\script{M}_{k-1}$ to its side information:  
\[ X_k + \script{M}_{k-1} = X_k + (X_{k-1} + X_k) = X_{k-1} \]
Finally, user $1$ performs the following computation (using the fact that it already has packet $X_1$): 
\begin{eqnarray*}
&& \hspace{-.3in} X_1 + \script{M}_1 + \script{M}_2 + \ldots + \script{M}_{K-1}   \\
&=& X_1 + (X_1 + X_2) + (X_2 + X_3) + \ldots + (X_{K-1} + X_K) \\
 &=& (X_1 + X_1) + (X_2+X_2) + \ldots + (X_{K-1} + X_{K-1}) \\
 && + X_K  \\
 &=& X_K
\end{eqnarray*}

Thus, such an operation can deliver $K$ packets in only $K-1$ transmissions. 
We call such an action a \emph{$K$-cycle coding action}.  We define a \emph{1-cycle coding action} to be a direct transmission.   
Note that $2$-cycle coding actions are the most ``efficient,'' having a packet/transmission efficiency ratio 
of $2/1$, compared to $K/(K-1)$ for $K\geq2$, 
which approaches 1 (the efficiency of a direct transmission) as $K\rightarrow\infty$.  While it is generally sub-optimal to restrict
to such cyclic coding actions, doing so can still provide significant gains in comparison to direct transmission.  Further, we show in 
Section \ref{section:broadcast} that such actions are optimal for certain classes of broadcast relay problems. 

Another important type of code action takes advantage of ``double-cycles'' in $\script{G}$:  Suppose for example
that user 1 wants packet $A$ and has packets $B$ and $C$, user 2 wants packet $B$ and has packets $A$ and $C$, and 
user 3 wants packet $C$ and has packets $A$ and $B$.  Then these demands can be fulfilled with the single transmission 
$A + B +C$, being a binary XOR of packets $A$, $B$, $C$.  The efficiency ratio of this action is $3/1$.

 \section{Dynamic Index Coding} \label{section:dynamic} 
 
 Now consider a dynamic setting where the broadcast station randomly receives packets from $M$ traffic
 flows.  Each flow $m \in \{1, \ldots, M\}$ contains fixed-length packets that must be delivered to a subset $\script{N}_m$ of the 
 users, and these packets are contained as side-information in a subset $\script{S}_m$ of the users.  We assume
 $\script{N}_m \cap \script{S}_m = \phi$, since any user $n \in \script{N}_m$ who wants the packet clearly does not already have
 the packet as side information.  
 In the general case, $M$ can be the number of all possible disjoint subset pair combinations.
 However, typically the value of $M$ will be much smaller than this, such as when each traffic flow represents packets from 
 a very large file, and there are only $M$ active file requests.  

Assume time is slotted with unit slots $t \in \{0, 1, 2,\ldots\}$, and let $\bv{A}(t) = (A_1(t), \ldots, A_M(t))$ be the number
of packets that arrive from each flow on slot $t$.  For simplicity of exposition, we assume the vector $\bv{A}(t)$ is i.i.d. 
over slots with expectation: 
\[ \expect{\bv{A}(t)} = \bv{\lambda} = (\lambda_1, \ldots, \lambda_M) \]
where $\lambda_m$ is the arrival rate of packets from flow $m$, in units of packets/slot. For simplicity, we assume that $A_m(t) \in \{0,1\}$ 
for all $m$ and all $t$, so that at most one new packet can arrive per flow per slot.   This is reasonable because the maximum delivery
rate in the system is one packet per slot, and so any packets that arrive as a burst can be ``smoothed'' and delivered to the network layer
at the broadcast station one slot at a time.   Packets of each flow $m$ are stored in a separate queue kept at the broadcast station, and exit
the queue upon delivery to their intended users.  

We now segment the timeline into  frames, each frame consisting of an integer number of slots.  At the beginning of each frame $r$, 
the network controller chooses a coding action $\alpha[r]$ 
within an abstract set $\script{A}$ of possible actions.  For each $\alpha \in \script{A}$, there is a \emph{frame size} 
$T(\alpha)$ and a \emph{clearance vector} $\bv{\mu}(\alpha)$.  The frame size $T(\alpha)$ is the number of slots
required to implement action $\alpha$, and is assumed to be a positive integer. The clearance vector
$\bv{\mu}(\alpha)$ has components $(\mu_1(\alpha), \ldots, \mu_M(\alpha))$, where $\mu_m(\alpha)$ is the number 
of type $m$ packets delivered as a result of action $\alpha$.  We assume $\mu_m(\alpha)$ is a non-negative
integer. When  frame $r$ ends, a new frame starts and the controller chooses a (possibly new) action $\alpha[r+1] \in \script{A}$. 
We assume each coding action only uses packets that are delivered as a result of that action, so that there
is no ``partial information''  that can be exploited on future frames.  
We further assume there are a finite (but arbitrarily large) number of coding actions in the set $\script{A}$, and that
there are positive numbers $T_{max}$ and $\mu_{max}$ such that $1 \leq T(\alpha) \leq T_{max}$ and
$0 \leq \mu(\alpha) \leq \mu_{max}$ for all $\alpha \in \script{A}$.  

Assume that frame $0$ starts at time $0$.  Define $t[0] = 0$, and for $r \in \{0, 1, 2, \ldots\}$ 
define $t[r]$ as the slot that starts frame $r$. 
 Let $\bv{Q}[r] = (Q_1[r], \ldots, Q_M[r])$ be the queue backlog vector at the beginning
of each frame $r \in \{0, 1, 2, \ldots\}$.  Then: 
\begin{equation} \label{eq:q-update} 
 Q_m[r+1] = \max[Q_m[r] - \mu_m(\alpha[r]), 0] + arrivals_m[r] 
 \end{equation} 
where $arrivals_m[r]$ is the number of type $m$ arrivals during frame $r$: 
\begin{equation} \label{eq:arrivals-m}
 arrivals_m[r] \defequiv \sum_{\tau=t[r]}^{t[r]+T(\alpha[r]) -1}A_m(\tau) 
 \end{equation} 
The $\max[\cdot, 0]$ operator in the queue update equation \eqref{eq:q-update} in principle allows actions $\alpha[r] \in \script{A}$
to be chosen independently of the queue backlog at the beginning of a frame.  In this case, if the action $\alpha[r]$ attempts
to deliver one or more packets from queues that are empty, \emph{null} packets are created and delivered.  In practice, these
null packets do not need to be delivered.

Our focus is on index coding problems with action sets $\script{A}$ defined by a 
specific set coding options, such as the set of all 
cyclic coding actions.  For example, an action $\alpha$ that is a 2-cyclic coding action that uses packets of  type $m$ and $k$ 
has $T(\alpha) = 1$ and $\bv{\mu}(\alpha)$ being a binary vector with 1s in entries $m$ and $k$ and zeros elsewhere. 
However, the above model is general and can also apply to other types of problems, 
such as multi-hop networks where actions $\alpha \in \script{A}$ represent some sequence of multi-hop network coding. 

\subsection{The Code-Constrained Capacity Region} 

We say that queue $Q_m[r]$ is \emph{rate stable} if: 
\[ \lim_{R\rightarrow\infty} \frac{Q_m[R]}{R} = 0 \: \: (\mbox{with probability 1}) \]
It is not difficult to show that $Q_m[R]$ is rate stable if and only if the arrival rate $\lambda_m$ is equal to the delivery
rate of type $m$ traffic \cite{sno-text}. The \emph{code-constrained capacity region} $\Lambda_{\script{A}}$ is the set of all rate 
vectors $(\lambda_1, \ldots, \lambda_M)$ for which there exists an algorithm for selecting $\alpha[r] \in \script{A}$ over
frames that makes all queues rate stable. 

\begin{thm} \label{thm:decomp} A rate vector $\bv{\lambda}$ is in the code-constrained capacity region $\Lambda_{\script{A}}$ if and only if
there exist  probabilities  $p(\alpha)$ such that $\sum_{\alpha\in\script{A}} p(\alpha) = 1$ and: 
\begin{eqnarray} \label{eq:cond1}  
\lambda_m \leq \frac{\sum_{\alpha\in\script{A}}p(\alpha)\mu_m(\alpha)}{\sum_{\alpha\in\script{A}} p(\alpha)T(\alpha)}  \: \: \: \: \forall m \in \{1, \ldots, M\}
\end{eqnarray} 
\end{thm} 

\begin{proof} 
The proof that such probabilities $p(\alpha)$ necessarily exist whenever $\bv{\lambda} \in \Lambda_{\script{A}}$ is 
given in Appendix B.  Below we prove sufficiency.  Suppose such probabilities $p(\alpha)$ exist that satisfy \eqref{eq:cond1}. 
We want to show that $\bv{\lambda} \in \Lambda_{\script{A}}$.  To do so, we design an algorithm that makes all queues $Q_m[r]$ in \eqref{eq:q-update}
rate stable.  
By rate stability theory in \cite{sno-text}, it suffices to design an algorithm that has a frame average arrival rate to each queue 
$Q_m[r]$
that is  less than or equal to the frame average service rate (both in units of packets/frame). 

Consider the algorithm that, every frame $r$, independently chooses action $\alpha \in \script{A}$ with probability $p(\alpha)$. 
Let $\alpha^*[r]$ represent this random action chosen on frame $r$. 
Then $\{T(\alpha^*[r])\}_{r=0}^{\infty}$ is an i.i.d. sequence, as is $\{\mu_m(\alpha^*[r])\}_{r=0}^{\infty}$ for each $m \in \{1, \ldots, M\}$. 
By the law of large numbers, 
the frame average arrival rate $\overline{arrivals}_m$ and the frame average service $\overline{\mu}_m$
 (both in packets/frame)  are equal to the following with probability 1: 
\begin{eqnarray*}
\overline{\mu}_m &=& \expect{\mu_m(\alpha^*[r])} =  \mbox{$\sum_{\alpha\in\script{A}}$} p(\alpha) \mu_m(\alpha) \\
\overline{arrivals}_m &=& \lambda_m\expect{T(\alpha^*[r])} = \lambda_m\mbox{$\sum_{\alpha\in\script{A}}$} p(\alpha)T(\alpha) 
\end{eqnarray*}
We thus have for each $m \in \{1, \ldots, M\}$: 
\begin{eqnarray*}
\frac{\overline{arrivals}_m}{\overline{\mu}_m} = \frac{\lambda_m\sum_{\alpha\in\script{A}} p(\alpha)T(\alpha)}{\sum_{\alpha\in\script{A}} p(\alpha)\mu_m(\alpha)} \leq 1
\end{eqnarray*}
where the final inequality follows by \eqref{eq:cond1}.  
\end{proof}

\subsection{Max-Weight Queueing Protocols} 

Theorem \ref{thm:decomp} shows that all traffic can be supported by a \emph{stationary and randomized} algorithm 
that independently chooses actions $\alpha^*[r] \in \script{A}$ with probability distribution $p(\alpha)$. This does not
require knowledge of the queue backlogs.  However, computing probabilities $p(\alpha)$ that satisfy \eqref{eq:cond1} 
would require knowledge of the arrival rates $\lambda_m$, and is a difficult computational task even if these rates are known. 
We provide two \emph{dynamic} algorithms that use queue backlog information.  These can also be viewed as online computation
algorithms for computing probabilities $p(\alpha)$.  Both are similar in spirit to the max-weight approach to dynamic scheduling 
in \cite{tass-server-allocation},
but the variable frame lengths require a non-trivial extended analysis. 
  Our  first algorithm assumes knowledge of the arrival rates $\lambda_m$.
    
  {\bf{Max-Weight Code Selection Algorithm 1 (Known $\bv{\lambda}$):}}  At the beginning of each frame $r$, observe the queue backlogs
  $Q_m[r]$  and perform the following: 
  \begin{itemize} 
  \item Choose code action $\alpha[r] \in \script{A}$ as the maximizer of: 
  \begin{equation} \label{eq:mw-alg1} 
   \sum_{m=1}^M Q_{m}[r][\mu_{m}(\alpha[r]) - \lambda_{m}T(\alpha[r])] 
   \end{equation} 
  where ties are broken arbitrarily. 
  \item Update the queue equation via \eqref{eq:q-update}.  
  \end{itemize} 
  
  The next algorithm uses a ratio rule, and does not require knowledge of the rates $\lambda_m$: 
  
    {\bf{Max-Weight Code Selection Algorithm 2 (Unknown $\bv{\lambda}$):}}  At the beginning of each frame $r$, observe the queue backlogs
    $Q_m[r]$  and perform the following: 
  \begin{itemize} 
  \item Choose code action $\alpha[r] \in \script{A}$ as the maximizer of: 
  \begin{equation} \label{eq:mw-alg2} 
   \sum_{m=1}^MQ_{m}[r]\left[\frac{\mu_{m}(\alpha[r])}{T(\alpha[r])}\right]
   \end{equation} 
  where ties are broken arbitrarily. 
  \item Update the queue equation via \eqref{eq:q-update}.  
  \end{itemize}

   \begin{thm} \label{thm:maxweight} Suppose that $\bv{\lambda} \in \Lambda_{\script{A}}$. Then all queues are
  rate stable under either of the two algorithms above.  
  \end{thm}

  \begin{proof} 
See Appendix C. 
  \end{proof} 
  
  It can further be shown that if there is a value $\rho$ such that $0 \leq \rho < 1$, and if
   $\bv{\lambda} \in \rho\Lambda_{\script{A}}$, being a $\rho$-scaled version of $\Lambda_{\script{A}}$, then 
   both algorithms give average queue size $O(1/(1-\rho))$.  Thus, the average backlog bound  
   increases to infinity as the arrival rates are pushed closer to the boundary of the capacity region. 
This is proven in Appendix D. 
  
  Define $\tilde{\script{A}}$ as the action space that restricts to direct transmissions, 
  2-cycle code actions, 3-cycle code actions, and the 1-slot $A+B+C$ code action
  that exploits double cycles, as described in Section \ref{section:cyclic-codes}. 
  Algorithm 2 has a particularly simple implementation on action space $\tilde{\script{A}}$
  and when each packet has at most one destination. 
  Indeed, we note that cycles can be defined purely on the user set $\script{N}$, and any candidate
  cycle that involves a user-to-user part $i\rightarrow j$ should use a packet of commodity $m \in\{1, \ldots, M\}$
  that maximizes $Q_m[r]$ over all commodities $m$ that consist of packets intended for user $j$ and
  contained as side information at user $i$.  

 \subsection{Example Simulation for 3 Users} 
 
 \begin{figure}[htbp]
    \centering
    \includegraphics[height=2in, width=3.5in]{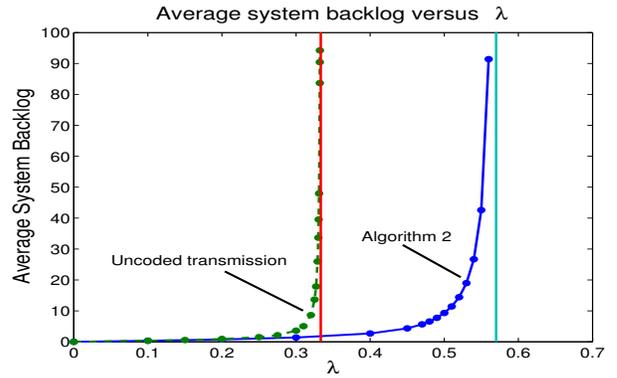} 
    \caption{Simulation of dynamic index coding for a 3 user system.}
    \label{fig:sim}
 \end{figure}
 
 Fig. \ref{fig:sim} presents simulation results for a system with  $N=3$ users, with action space $\tilde{\script{A}}$ as defined
 above. 
 We consider only algorithm 2, which does not require knowledge of rates $\lambda_m$, and 
 compare against uncoded transmissions. 
 All packets are intended for at most
 one user.  
 Packets intended for user $n \in \{1, 2, 3\}$ arrive as independent Bernoulli processes with identical 
 rates $\lambda$.  We assume each packet is independently in the cache of the other two users with probability 
 1/2.  Thus, there are four types of packets intended
 for user 1: Packets not contained as side information anywhere, packets contained as side information at user 2 only, 
 packets contained as side information at user 3 only, and packets contained as side information at both users 2 and 3. 
 Users 2 and 3 similarly have 4 traffic types, for a total of $M=12$ traffic types.  

Each data point in Fig. \ref{fig:sim} represents a simulation over 5 million frames at a given value of $\lambda$.  
The figure plots the resulting total average number of packets in the system (summed over all 12 queues). 
The case of direct (uncoded) transmission is also shown.  Uncoded transmission can support a 
maximum rate of $\lambda = 1/3$  (for a total traffic rate of 1). 
It is seen that algorithm 2 can significantly outperform uncoded transmission, achieving stability 
at rates up to $\lambda = 0.57$ (for a total traffic rate of $1.71$). 

\section{Broadcast Relay Networks} \label{section:broadcast}

\begin{figure}[t]
   \centering
   \includegraphics[height=1in, width=3in]{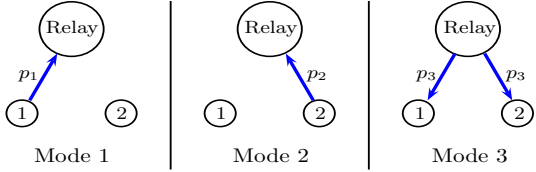} 
   \caption{An illustration of a 2-user broadcast relay system, with the 3 possible transmission modes shown. In mode 3, 
   packet $p_3$ is received at both users.}
   \label{fig:1}
\end{figure}

Consider now the following related problem:  There are again $N$ users and a single broadcast station. 
However, the broadcast station initially has no information, and acts as a relay to transfer independent
unicast 
data between the users.  Further, the users only know their own data, and initially have no knowledge of data sourced at other
users.  Time is again slotted, and every slot we can choose
from one of $N+1$ modes of transmission.  The first $N$ transmission modes involve
an error-free packet transmission from a single user to the relay.  
The $(N+1)$th transmission mode is where the relay broadcasts 
a single packet that is received error-free at each of the $N$ users. 
Fig. \ref{fig:1} illustrates
an example system with 2 users, where the 3 possible transmission modes are shown. 
For simplicity, we assume the user transmissions cannot be overheard by other users, and the 
users first send all packets to the relay.  The relay then can make coding decisions for its 
downlink transmissions. 


\subsection{The Minimum Clearance Time Relay Problem} 

First consider a static problem where a batch of packets must be delivered in minimum time. 
Let $P_{ij}$ represent the number of packets that user $i$ wants to send to user $j$, where
$i, j \in \{1, \ldots, N\}$. All packets are independent, and the total number of packets is $P$, where:
\[ P = \mbox{$\sum_{i=1}^N\sum_{j=1}^N$} P_{ij} \]
This problem is related to the index coding problem as follows:   Suppose on the first $P$ slots, all users
send their packets to the relay on the uplink channels.  It remains for the relay to send all users the desired data, 
and these users have side information.  The resulting side information graph $\script{G}$ is the same as in the general 
index coding problem.  However, it has the following \emph{special structure}: The only user that has side information about a packet
is the source user of the packet.  Specifically: 
\begin{itemize} 
\item Each packet is contained as side information in exactly one user.  Thus, 
each packet node of $\script{G}$ has a single incoming link from some user that is its source.
\item Each packet has exactly one user as its destination.  
Thus, each packet node of $\script{G}$ has a single outgoing link to some user that is its destination.
\end{itemize} 

This special structure leads to a simplified graphical model for demands, which we call the \emph{weighted compressed
graph} $\WCG$ of $\script{G}$. The graph $\WCG$ is formed from $\script{G}$ as follows:  It is a directed graph defined
on the user nodes $\script{N}$ only, and contains a link $(a,b)$ if and only if the original graph $\script{G}$ specifies that
user node $a$ has a packet that user node $b$ wants.  Further, each link $(a,b)$ is given a positive integer weight $P_{ab}$, 
the number of packets user $a$ wants to send to user $b$.  It is easy to show that $\WCG$ is acyclic if and only if $\script{G}$ is 
acyclic (see Appendix E).  Hence, coding can only help if $\WCG$ contains cycles, and so $T_{min}(\script{G}) = P$ whenever $\WCG$ is
acyclic.  

We say the weighted compressed graph $\WCG$ has \emph{disjoint cycles} if each link participates in at most one simple
cycle. An example is shown in Fig. \ref{fig:disjoint-cycles}.   Consider such a graph that has $C$ disjoint cycles.  Let $w^{min}_c$ be 
the min-weight link on each disjoint cycle $c \in \{1, \ldots, C\}$.  

\begin{thm} \label{thm:disjoint-cycles} If the broadcast relay problem has a weighted compressed graph $\WCG$ with disjoint cycles, then: 
\begin{equation} \label{eq:disjoint-cycles} 
 T_{min}(\script{G}) = P - \mbox{$\sum_{c=1}^Cw^{min}_c$} 
 \end{equation} 
and so the full clearance time (including the $P$ uplink transmissions) is the above number plus $P$. 
Further, optimality can be achieved over the class of cyclic coding actions, as described in Section \ref{section:cyclic-codes}. 
\end{thm} 

As an example, the graph $\WCG$ in Fig. \ref{fig:disjoint-cycles}a has $P= 48$, three disjoint
cycles with $w^{min}_1 = 4$, $w^{min}_2 = 4$, $w^{min}_3 = 1$, and so $T_{min}(\script{G}) = 48 - 4 - 4 - 1 = 39$. 

\begin{figure}[htbp]
   \centering
   \includegraphics[height=1.1in, width=3.25in]{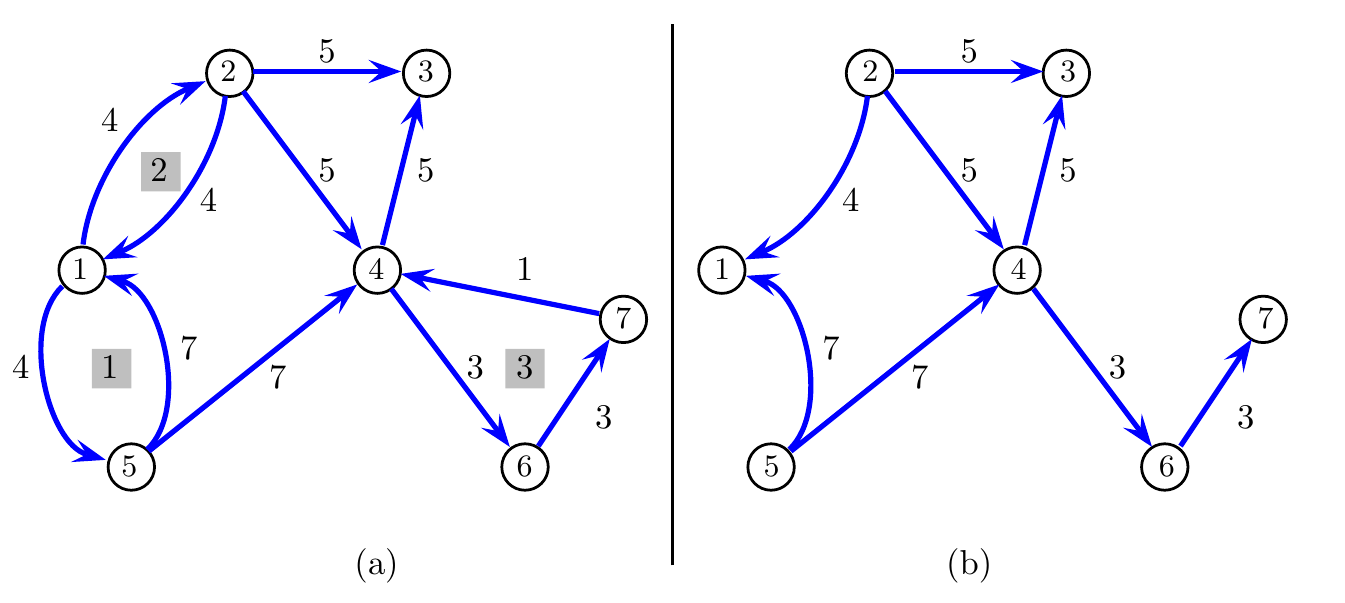} 
   \caption{(a) A graph $\WCG$ with three disjoint cycles, and (b) its pruned graph $\script{W}\script{C}(\script{G}')$.}
   \label{fig:disjoint-cycles}
\end{figure}

\begin{proof} 
First prune the graph 
$\WCG$ by removing the min-weight link on each of the disjoint cycles (breaking ties arbitrarily).  This corresponds to removing those packets 
from the original graph $\script{G}$, to produce a new graph $\script{G}'$ with exactly $P - \sum_{c=1}^Cw_c^{min}$ packets. 
The weighted compressed graph $\script{W}\script{C}(\script{G}')$ is the subgraph of $\WCG$ with the min-weight links
on each disjoint cycle removed (see Fig. \ref{fig:disjoint-cycles}a and Fig. \ref{fig:disjoint-cycles}b). Both $\script{G}'$ and $\script{W}\script{C}(\script{G}')$ are acyclic, and so: 
\[ T_{min}(\script{G}) \geq T_{min}(\script{G}') = P - \mbox{$\sum_{c=1}^Cw_c^{min}$}  \]
It remains only to construct a coding algorithm that achieves this lower bound.  This can be done easily by using $w_c^{min}$ separate 
cyclic coding
actions for each of the disjoint cycles (using a $k$-cycle coding action for any cycle of length $k$), 
and then directly transmitting the remaining packets. 
\end{proof} 

\subsection{Traffic Structure and Optimality of Cyclic Coding}

Suppose we have a broadcast relay problem with $N$ users, packet matrix $(P_{ij})$, 
and with the following additional structure:  Each user $i \in \{1, \ldots, N\}$ wants to send data to only one other user.
That is, the matrix $(P_{ij})$ has at most one non-zero entry in each row $i \in \{1, \ldots, N\}$.   
We now show that the resulting graph $\WCG$ has disjoint cycles.  To see this, suppose it is not true, so that there are two overlapping cycles. Then there must be a shared
   link $(a,b)$ that continues to a link $(b, k)$ for cycle 1 and $(b, m)$ for cycle 2, where $k \neq m$.  This 
   means node $b$ has two outgoing links,  a contradiction because matrix $(P_{ij})$ has at most one non-zero entry in 
   row $b$,  and hence at most one outgoing link from node $b$.    
   We conclude that $\WCG$ has disjoint cycles, and so cyclic coding is optimal via Theorem \ref{thm:disjoint-cycles}. 
   
  A similar argument holds if each user 
  wants to \emph{receive} from at most one other user, so that $(P_{ij})$ has at most one non-zero entry 
  in every column.   Again, $\WCG$ has disjoint cycles, and so cyclic coding is optimal.

\subsection{Dynamic Broadcast Relay Scheduling} 

Now consider the dynamic case where packets from source user $i$ and destination user $j$ arrive with rate $\lambda_{ij}$
packets/slot. Suppose we have an abstract set of coding actions $\script{A}$, where
each action involves a subset of packets, and first transmits these packets to the relay before any coding at the relay.  
Let $T(\alpha)$ be the 
number of slots to complete the action, and $(\mu_{ij}(\alpha))$ be the matrix of packets delivered by the action. 
It can be shown that capacity can be approached arbitrarily closely by repetitions of minimum-clearance time scheduling on 
large blocks of the incoming data (similar to the capacity treatment in \cite{nonlinear-index-coding} for a limit of large packet size).  Hence, if $T_{min}(\script{G})$ can be optimally solved using only cyclic-coding actions, 
then capacity is also achieved in the max-weight algorithms when $\script{A}$ is restricted to cyclic-coding actions. 
It follows that such actions are optimal for rate matrices $(\lambda_{ij})$ with at most one non-zero entry per row, and for 
 rate matrices $(\lambda_{ij})$ with at most one non-zero entry per column. 

\subsection{Counterexamples} 

$\vspace{-.3in}$

\begin{figure}[htbp]
   \centering
   \includegraphics[height=1.1in, width=3.4in]{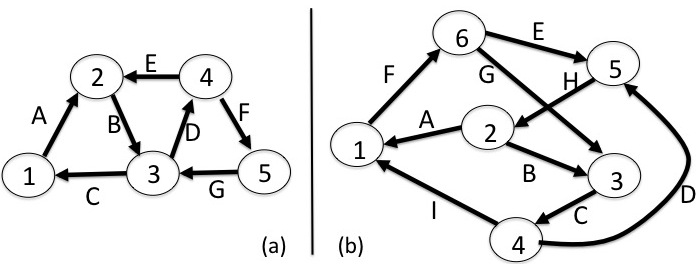} 
   \caption{Two example graphs $\WCG$ for broadcast relay problems.}
   \label{fig:counter-example}
\end{figure}

Can we minimize clearance time by grabbing any 
available 2-cycle, then any available 3-cycle if no 2-cycle
is available, and so on?  Not necessarily.  A simple counterexample is shown in Fig. \ref{fig:counter-example}a. 
%
The graph has 
5 users and 7 packets $\{A, \ldots, G\}$, where each link has a single packet.  Using the 
middle 3-cycle $2\rightarrow 3 \rightarrow 4\rightarrow 2$ by transmitting $B + D$ and $D + E$ leaves a remaining acyclic graph with 4 packets, and hence would take
4 more transmissions, for a total of 6 slots.  However, using the two side cycles (with 2 transmissions each) and then
transmitting the remaining packet $E$ clears everything in 5 slots, which is optimal because the maximum acyclic
subgraph has 5 packets (just remove links $B$ and $D$). 

One may wonder if all broadcast relay graphs can be optimally cleared with cyclic coding.  We can
show this is 
true for $N=2$ and $N=3$ (see Appendix E and F).  However, this is not true in general for $N>3$. Fig. \ref{fig:counter-example}b 
shows a counterexample with $N=6$. Suppose each link has a single packet, 
so that we have 9 packets $\{A, \ldots, I\}$.  It can be shown that the maximum acyclic subgraph
has 7 packets, and so $T_{min}(\script{G}) \geq 7$, but the best cyclic coding method uses 8 slots.  
Here is a way to achieve 7 slots:  Send messages
$M_1 = E + G + F$, $M_2 = H + E$, $M_3 = H + D$, $M_4 = A + B + H$, $M_5 = C + B$, $M_6 = C + G$, $M_7 = C + I + D$. 
The decodings at users $2, 3, 4, 5, 6$ are straightforward by combining their side information with just a single message.  
The decoding at user 1 is done as follows:  $M_1 + M_2 + M_3 + M_6 + M_7 = F + I$. 
Since user 1 knows $F$, it can decode $I$.  $M_3 + M_4 + M_5+M_7 = A + I$, since user 1 knows $I$ it can get $A$. 

\section{Conclusions} 

This work presents a dynamic approach to index coding.  This problem is important for future wireless communication, 
where there may be many instances of side information that can be exploited. While optimal index coding for general 
problems seems
to be intractable, we develop a code-constrained capacity region, 
which restricts actions to a pre-specified set of codes.  Two max-weight algorithms were developed that can support
randomly arriving traffic whenever the arrival rate vector is inside the code-constrained capacity region.  
The first algorithm requires knowledge of the rate vector, and the second does not.   
Simulations 
verify network stability up to the boundary of the code-constrained capacity region, and illustrate 
improvements in both throughput and delay over uncoded transmission. 
For coding to provide gains in comparison to 
direct transmission, it must exploit cycles in the demand graph.  A simple set of codes based on cycles was
considered and shown to be optimal (so that the code constrained capacity region is equal to the unconstrained
capacity region) for certain classes of broadcast relay networks.  
These results add to the theory of information networks, and 
can be used to improve efficiency
in wireless communication systems.

\section*{Appendix A --- Proof of Theorem \ref{thm:acyclic}} 

 \begin{proof} (Theorem \ref{thm:acyclic}) 
 We already know that $T_{min}(\script{G}) \leq P$.  It suffices to show that $T_{min}(\script{G}) \geq P$. 
 Consider any mission-completing coding action that takes $T$ slots.  We show that $T \geq P$. 
 Let $\script{M}$ be the sequence of messages transmitted.  Then every node $n \in \script{N}$ 
 is able to decode its desired packets, being packets in the set $\script{R}_n$, 
 from the information $\{\script{H}_n, \script{M}\}$, being the 
 information it has at the end of the coding action.  That is, we have: 
 \begin{equation} \label{eq:decode-equiv-app}
  \{\script{H}_n, \script{M}\} \iff \{\script{H}_n, \script{M}, \script{R}_n\} \: \: \forall n \in\{1, \ldots, N\} 
 \end{equation} 
 
 Because the graph is acyclic, there must be at least one node with no outgoing links (by Fact \ref{fact:leaf}).  Choose
 such a node, and label this node $n_1$.   The node $n_1$ 
 cannot be a packet node, because we have assumed that all packet nodes have outgoing links.  Thus, $n_1 \in \script{N}$. 
 Because node $n_1$ has no outgoing links, it has $\script{H}_{n_1} = \phi$ and thus has
 no initial side information about any of the packets.  Thus, 
 it is able to decode all packets in the set $\script{R}_{n_1}$ by the messages $\script{M}$ alone.  That is: 
 \begin{equation} \label{eq:base} 
  \script{M} \iff \{\script{M}, \script{R}_{n_1}\} 
 \end{equation} 
 We want to show that this node $n_1$ can decode \emph{all} packets in the set $\script{P}$, so that: 
 \begin{equation} \label{eq:claim} 
 \script{M} \iff \{\script{M}, \script{P} \} 
 \end{equation} 
  If we can show that \eqref{eq:claim} holds, then the sequence of messages $\script{M}$ is also sufficient to deliver 
  $P$ independent packets to 
  node $n_1$, and node $n_1$ did not have any initial side information about these packets. 
  Thus, the number of slots $T$ used in the coding action must be at least $P$ by Fact \ref{fact:onelink}, proving the 
  result.  Thus, it suffices to prove \eqref{eq:claim}. 
  
  We prove \eqref{eq:claim} by induction on $k$, for $k \in \{1, \ldots, N-1\}$:  Assume that there is a
labeling of $k$ distinct user nodes $\{n_1, n_2, \ldots, n_k\}$ such that: 
\begin{equation} \label{eq:induction} 
 \{\script{M} \} \iff \{\script{M}, \script{R}_{n_1}, \ldots, \script{R}_{n_k}\} 
 \end{equation} 
This property holds for the base case $k=1$ by \eqref{eq:base}.  We now assume that \eqref{eq:induction} holds
for a general $k \in \{1, \ldots, N-1\}$, and prove it must also hold for $k+1$.  Take the graph $\script{G}$, and delete
the user nodes $\{n_1, \ldots, n_k\}$, also deleting all links outgoing from and incoming to these nodes.  This may create
packet nodes with no outgoing links:  Delete all such packet nodes.  Note that all deleted packet nodes (if any) 
must be in the set $\{\script{R}_{n_1}, \ldots, \script{R}_{n_k}\}$, being the set of packets desired by the users that are deleted.  
The resulting
subgraph must still be acyclic, and hence it must have a node $n_{k+1}$ with no outgoing links. This node must
be a user node, as we have deleted all packet nodes with no outgoing links.  

 Because the user node $n_{k+1}$ has no outgoing links, it either had $\script{H}_{n_{k+1}} = \phi$ (so that it never had any
 outgoing links), or all of its outgoing links were pointing to packet nodes that we have deleted, and so those
 packets were in the set $\{\script{R}_{n_1}, \ldots, \script{R}_{n_k}\}$.  That is, 
 we must have $\script{H}_{n_{k+1}} \subseteq \{\script{R}_{n_1}, \ldots, \script{R}_{n_k}\}$.  
  Therefore: 
  \begin{equation} \label{eq:have} 
   \{ \script{M}, \script{H}_{n_{k+1}}\} \subseteq \{\script{M}, \script{R}_{n_1}, \ldots, \script{R}_{n_k}\} 
   \end{equation} 
     However, at the end of the coding action, 
     node $n_{k+1}$ has exactly the information on the left-hand-side of \eqref{eq:have}, and 
     hence this information is sufficient to decode all packets in the set $\script{R}_{n_{k+1}}$.  Thus, the information 
     on the right-hand-side of \eqref{eq:have} must also be sufficient to decode $\script{R}_{n_{k+1}}$, so that: 
     \[   \{\script{M}, \script{R}_{n_1}, \ldots, \script{R}_{n_k}\} \iff  \{\script{M}, \script{R}_{n_1}, \ldots, \script{R}_{n_k}, \script{R}_{n_{k+1}}\}  \]
  But this together with \eqref{eq:induction} yields: 
  \[  \{\script{M} \} \iff \{\script{M}, \script{R}_{n_1}, \ldots, \script{R}_{n_k}, \script{R}_{n_{k+1}}\} \] 
which completes the induction step. 

By induction over $k \in \{1, \ldots, N-1\}$, it follows that: 
  \begin{equation} \label{eq:foo} 
    \{\script{M} \} \iff \{\script{M}, \script{R}_{n_1}, \ldots, \script{R}_{n_{N}}\} 
    \end{equation}  
  However, by re-labeling we have: 
  \begin{equation} \label{eq:foofoo}
   \{ \script{R}_{n_1}, \ldots, \script{R}_{n_{N}}\} = \{\script{R}_1, \ldots, \script{R}_N\} = \script{P} 
   \end{equation} 
  where the final equality holds by \eqref{eq:P}.  Combining \eqref{eq:foo} and \eqref{eq:foofoo} proves \eqref{eq:claim}. 
   \end{proof}

\section*{Appendix B -- Proof of Necessity for Theorem \ref{thm:decomp}} 

Let $\{\alpha[r]\}_{r=0}^{\infty}$ be a sequence of actions, chosen over frames,  that makes all queues $Q_m[r]$ rate
stable.  We show there must exist probabilities $p(\alpha)$ that satisfy \eqref{eq:cond1}. 
For each positive integer $R$ and each $m \in \{1, \ldots, M\}$,  
define $\overline{a}_m[R]$ and $\overline{\mu}_m[R]$ as the following averages over the first $R$ frames: 
\begin{eqnarray*}
\overline{a}_m[R] &\defequiv& \mbox{$\frac{1}{R}\sum_{r=0}^{R-1}$}arrivals_m[r] \\
\overline{\mu}_m[R] &\defequiv& \mbox{$\frac{1}{R}\sum_{r=0}^{R-1}$} \mu_m(\alpha[r])
\end{eqnarray*}
where $arrivals_m[r]$ is defined in \eqref{eq:arrivals-m}.   Now define $\script{F}(\alpha,R)$ as the 
set of frames $r \in \{0, \ldots, R-1\}$ that use action $\alpha$, and define $|\script{F}(\alpha, R)|$ as the 
number of these frames, so that $\sum_{\alpha\in\script{A}} |\script{F}(\alpha, R)| = R$. 
We then have: 
\begin{eqnarray}
\overline{a}_m[R] &=& \mbox{$\sum_{\alpha\in\script{A}}$} \frac{|\script{F}(\alpha, R)|}{R} \times \nonumber \\
&& \frac{1}{|\script{F}(\alpha, R)|}\sum_{r\in\script{F}(\alpha, R)} arrivals_m[r] \label{eq:a-decomp-app} \\
\overline{\mu}_m[R] &=& \mbox{$\sum_{\alpha\in\script{A}}$}\frac{|\script{F}(\alpha,R)|}{R}\mu_m(\alpha) \label{eq:mu-decomp-app} 
\end{eqnarray}
The set $\script{A}$ is finite.  Thus,  the values $\{(|\script{F}(\alpha, R)|/R)\}_{R=1}^{\infty}$ can be viewed as an infinite sequence 
of bounded vectors 
(with dimension equal to the size of set $\script{A}$) defined on the index $R \in \{1, 2, 3, \ldots\}$, and hence
must have a convergent subsequence.  Let $R_k$ represent the sequence of frames on this subsequence, so that
there are values $p(\alpha)$ for all $\alpha \in \script{A}$ such that: 
\begin{eqnarray*}
\lim_{k\rightarrow\infty} |\script{F}(\alpha, R_k)|/R_k = p(\alpha)
\end{eqnarray*}
Further, by \eqref{eq:mu-decomp-app} we have for all $m \in \{1, \ldots, M\}$: 
\begin{equation} \label{eq:mu-limit} 
 \lim_{k\rightarrow\infty} \overline{\mu}_m[R_k] = \mbox{$\sum_{\alpha\in\script{A}}$} p(\alpha)\mu_m(\alpha) 
 \end{equation} 
 Likewise, from \eqref{eq:a-decomp-app} and the law of large numbers (used over each $\alpha\in\script{A}$ for which 
 $\lim_{k\rightarrow\infty} |\script{F}(\alpha,R_k)| = \infty$, and noting that $arrivals_m[r]$ is i.i.d. with mean $T(\alpha)\lambda_m$
 for all $r \in \script{F}(\alpha, R)$) we have with probability 1: 
\begin{equation} \label{eq:a-limit} 
\lim_{k\rightarrow\infty} \overline{a}_m[R_k] = \sum_{\alpha\in\script{A}} p(\alpha)T(\alpha)\lambda_m
\end{equation} 

Because $|\script{F}(\alpha, R_k)|/R_k \geq 0$ for all $\alpha\in\script{A}$ and all $R_k$, 
and  $\sum_{\alpha\in\script{A}} |\script{F}(\alpha, R_k)|/R_k = 1$ for all $R_k$, 
the same holds for the limiting values $p(\alpha)$.  That is, 
$p(\alpha) \geq 0$ for all $\alpha \in \script{A}$, and: 
$\sum_{\alpha\in\script{A}} p(\alpha) = 1$.
Because each queue $Q_m[r]$ is rate stable, we have with probability 1 that for all $m \in \{1, \ldots, M\}$: 
\begin{equation} \label{eq:Qrs} 
 \lim_{k\rightarrow\infty} \frac{Q_m[R_k]}{R_k} = 0 
 \end{equation} 
However, from the queue update equation \eqref{eq:q-update} we have for all $r \in \{0, 1, 2, \ldots\}$: 
\[ Q_m[r+1] \geq Q_m[r] - \mu_m(\alpha[r]) + arrivals_m[r] \]
Summing the above over $r \in \{0, 1, \ldots, R_k-1\}$ and dividing by $R_k$ yields: 
\[ \frac{Q_m[R_k] - Q_m[0]}{R_k} \geq -\overline{\mu}_m[R_k] + \overline{a}_m[R_k] \]
Taking a limit as $k\rightarrow\infty$ and using \eqref{eq:mu-limit}-\eqref{eq:Qrs} yields:  
\begin{equation} \label{eq:penultimate-app2} 
0 \geq \mbox{$-\sum_{\alpha\in\script{A}}$} p(\alpha)\mu_m(\alpha) + \lambda_m\sum_{\alpha\in\script{A}} p(\alpha)T(\alpha)
\end{equation} 
This proves the result. 

 \section*{Appendix C --- Proof of Theorem \ref{thm:maxweight}} 
  
  We first prove rate stability for Algorithm 2, which uses a ratio rule. 
  The proof for Algorithm 1 is simpler and is given after. 
  We have the following preliminary lemma. 
    
  \begin{lem} \label{lem:suff}  (Sufficient Condition for Rate Stability \cite{lyap-opt2}): 
Let $Q[r]$ be a non-negative stochastic process defined over the integers $r \in \{0, 1, 2,  \ldots\}$. Suppose
there are constants $B$, $C$, $D$ such that for all frames $r \in \{0, 1, 2,  \ldots\}$ we have: 
\begin{eqnarray}
\expect{(Q[r+1] - Q[r])^2} &\leq& D \label{eq:rs1}  \\
\expect{Q[r]^2} &\leq& Br + C \label{eq:rs2}  
\end{eqnarray}
Then $\lim_{r\rightarrow\infty}Q[r]/r = 0$ with probability 1.   
  \end{lem} 
 
  The condition \eqref{eq:rs1} is immediately satisfied in our system because the   
  queue changes over any frame are bounded. Thus, to prove rate stability, 
   it   suffices to show that \eqref{eq:rs2} holds for all queues and 
   all frames.  That is, it suffices to prove the second moment of queue
  backlog grows at most linearly.

 For each frame $r \in \{0, 1, 2, \ldots\}$, define the following 
 quadratic function $L[r]$, called a \emph{Lyapunov function}: 
  \[ L[r] \defequiv \mbox{$\frac{1}{2}\sum_{m=1}^M$}Q_{m}[r]^2 \]
  Define the \emph{conditional Lyapunov drift} $\Delta[r]$ to be the expected change in $L[r]$ from one frame to the
  next: 
  \[ \Delta[r] \defequiv \expect{L[r+1] - L[r] | \bv{Q}[r]} \]
  where $\bv{Q}[r] = (Q_1[r], \ldots, Q_M[r])$ is the queue backlog vector on frame $r$.  The above 
 conditional expectation is with respect to the random arrivals over the frame and the (possibly random)
  coding action chosen for the frame.

  \begin{lem} \label{lem:causal-drift} 
  Under any (possibly randomized) decision for $\alpha[r] \in \script{A}$ that is \emph{causal} (i.e., that
  does not know the future values of arrivals over the frame), we have for each frame $r$: 
   \begin{eqnarray*}
   \Delta[r] \leq B  + \mbox{$\sum_{m=1}^M$}Q_{m}[r]\expect{\lambda_{m}T(\alpha[r]) - \mu_{m}(\alpha[r])|\bv{Q}[r]} 
   \end{eqnarray*}
   where $B$ is a finite constant that satisfies: 
   \[ B \geq \mbox{$\frac{1}{2}\sum_{m=1}^M$}\expect{arrivals_m[r]^2 + \mu_m(\alpha[r])^2|\bv{Q}[r]} \]
   Such a finite constant $B$ exists because frame sizes are bounded, as are the arrivals per slot. 
  \end{lem} 
 \begin{proof} 
 For simplicity of notation, define $b_m[r] \defequiv \mu_{m}(\alpha[r])$, and 
  $a_m[r] \defequiv arrivals_m[r]$. 
 The queue update equation is thus: 
 \[ Q_m[r+1] = \max[Q_m[r] - b_m[r], 0] + a_m[r] \]
 Note that for any non-negative values $Q, a, b$ we have: 
  \[ (\max[Q - b, 0] + a)^2 \leq 
 Q^2 + b^2 + a^2 + 2Q(a-b) \]
 Using this and 
 squaring the queue update equation yields: 
 \[ Q_{m}[r+1]^2 \leq Q_{m}[r]^2 + b_{m}[r]^2 + a_{m}[r]^2 + 2Q_{m}[r][a_m[r] - b_m[r]] \]
 Summing over all $m$, dividing by $2$, and taking conditional expectations yields: 
 \begin{equation} \label{eq:last} 
  \Delta[r] \leq B + \mbox{$\sum_{m=1}^M$}Q_{m}[r]\expect{a_{m}[r] - b_{m}[r]|\bv{Q}[r]} 
  \end{equation} 
 Now note that:\footnote{Equality \eqref{eq:zappa} uses causality and the i.i.d. nature of the arrival process.
 It is formally proven by conditioning on $T(\alpha[r])$ and using iterated expectations.} 
 \begin{eqnarray}
 \expect{a_{m}[r]| \bv{Q}[r]} &=& \expect{\sum_{\tau={t[r]}}^{t[r]+T(\alpha[r])-1} A_m(\tau)|\bv{Q}[r]}  \nonumber \\
 &=& \expect{\lambda_mT[r]|\bv{Q}[r]} \label{eq:zappa}
 \end{eqnarray}
 Plugging this identity into \eqref{eq:last} proves the result. 
 \end{proof}

 We now prove that Algorithm 2 yields rate stability.

\begin{proof} (Theorem \ref{thm:maxweight}---Stability Under Algorithm 2) 
Suppose that Algorithm 2 is used, so that we choose $\alpha[r]$ every frame $r$ via \eqref{eq:mw-alg2}. 
We first claim that for each frame $r$ and for all possible $\bv{Q}[r]$ we have:  
\begin{eqnarray}
\frac{\expect{\sum_{m=1}^MQ_{m}[r]\mu_{m}(\alpha[r])|\bv{Q}[r]}}{\expect{T(\alpha[r])|\bv{Q}[r]}} \geq \nonumber \\
\frac{\expect{\sum_{m=1}^MQ_{m}[r]\mu_{m}(\alpha^*[r])|\bv{Q}[r]}}{\expect{T(\alpha^*[r])|\bv{Q}[r]}} \label{eq:ratio-compare}
\end{eqnarray}
where $\alpha^*[r]$ is any other (possibly randomized) code action that could be chosen over the options in the 
set $\script{A}$. 
This can be shown as follows:  Suppose we want to choose $\alpha[r] \in \script{A}$ via a possibly randomized
decision, to maximize the ratio of expectations in the left-hand-side of \eqref{eq:ratio-compare}.  Such a decision would
satisfy \eqref{eq:ratio-compare} by definition, since it would maximize the ratio of expectations over all alternative
policies $\alpha^*[r]$.  However, it is known that such a maximum is achieved via a \emph{pure policy} that chooses
a particular $\alpha \in \script{A}$ with probability 1 (see Chapter 7 of \cite{sno-text}).  The best pure policy is thus
the one that observes the queue backlogs $\bv{Q}[r]$ and chooses $\alpha[r] \in \script{A}$ to maximize the deterministic
ratio, which is exactly how Algorithm 2 chooses its action (see \eqref{eq:mw-alg2}).  

Thus, \eqref{eq:ratio-compare} holds.  We can rewrite \eqref{eq:ratio-compare} as: 
\begin{eqnarray}
\frac{\expect{\sum_{m=1}^MQ_{m}[r]\mu_{m}(\alpha[r])|\bv{Q}[r]}}{\expect{T(\alpha[r])|\bv{Q}[r]}} \geq \nonumber \\
\sum_{m=1}^MQ_{m}[r]\frac{\expect{\mu_{m}(\alpha^*[r])|\bv{Q}[r]}}{\expect{T(\alpha^*[r])|\bv{Q}[r]}} \label{eq:ratio-compare2} 
\end{eqnarray}
We can thus plug any alternative (possibly randomized) decision  $\alpha^*[r]$ into the right-hand-side
of \eqref{eq:ratio-compare2}.   Consider the randomized algorithm that independently selects $\alpha\in\script{A}$ every 
frame, independent of queue backlogs, according to the distribution $p(\alpha)$ in Theorem \ref{thm:decomp}.
Let $\alpha^*[r]$ represent the randomized decision under this policy.  Then from \eqref{eq:cond1} we have
for all $m \in \{1, \ldots, M\}$: 
\begin{eqnarray}
 \lambda_m \leq \frac{\expect{\mu_m(\alpha^*[r])}}{\expect{T(\alpha^*[r])}} =
 \frac{\expect{\mu_m(\alpha^*[r])|\bv{Q}[r]}}{\expect{T(\alpha^*[r])|\bv{Q}[r]}} \label{eq:again} 
 \end{eqnarray}
where the last equality holds because $\alpha^*[r]$ is chosen independently of $\bv{Q}[r]$. Using this
in \eqref{eq:ratio-compare2} yields: 
\begin{eqnarray*}
\frac{\expect{\sum_{m=1}^MQ_{m}[r]\mu_{m}(\alpha[r])|\bv{Q}[r]}}{\expect{T(\alpha[r])|\bv{Q}[r]}} \geq 
\sum_{m=1}^MQ_{m}[r]\lambda_{m} 
\end{eqnarray*}
Rearranging terms above yields: 
\begin{equation} \label{eq:appc-final} 
\mbox{$\sum_{m=1}^M$}Q_m[r]\expect{\lambda_mT(\alpha[r]) - \mu_m(\alpha[r])|\bv{Q}[r]} \leq 0 
\end{equation} 
Plugging \eqref{eq:appc-final} into the drift bound of Lemma \ref{lem:causal-drift} yields: 
\[ \Delta[r] \leq B \]
Taking expectations of the above and using the definition of $\Delta[r]$ yields: 
\[ \expect{L[r+1]} - \expect{L[r]} \leq B \: \: \forall r \in \{0, 1, 2, \ldots\} \]
Summing the above over $r \in \{0, 1, \ldots, R-1\}$  yields: 
\[ \expect{L[R]}  - \expect{L[0]}\leq  BR \]
and hence for all $R>0$: 
\[ \mbox{$\sum_{m=1}^M$}\expect{Q_m[R]^2} \leq 2\expect{L[0]} + 2BR \]
Thus, the second moments of all queues grow at most linearly, from which we 
guarantee rate stability by Lemma \ref{lem:suff}. 
\end{proof} 

 We now prove that Algorithm 1 yields rate stability. 
 
 \begin{proof} (Theorem \ref{thm:maxweight}---Stability Under Algorithm 1) 
 Note that Algorithm 1 is designed to observe queue backlogs $\bv{Q}[r]$ every frame $r$, and 
 take a control action $\alpha[r] \in \script{A}$ to minimize the right-hand-side of the drift bound 
 in Lemma \ref{lem:causal-drift}.  Therefore, we have: 
  \begin{eqnarray*}
  \Delta[r] \leq B + \sum_{m=1}^MQ_m[r]\expect{\lambda_mT(\alpha^*[r]) - \mu_m(\alpha^*[r])|\bv{Q}[r]}
  \end{eqnarray*}
 where $\alpha^*[r]$ is any other (possibly randomized) decision.  If $\alpha^*[r]$ makes a decision independent
 of $\bv{Q}[r]$ we have: 
 \begin{eqnarray}
   \Delta[r] \leq B + \sum_{m=1}^MQ_m[r]\expect{\lambda_mT(\alpha^*[r]) - \mu_m(\alpha^*[r])}
  \label{eq:alg1} 
\end{eqnarray}
 Consider again 
 randomized algorithm $\alpha^*[r]$ that independently and randomly selects an action in $\script{A}$ every 
frame, independent of queue backlogs, according to the distribution $p(\alpha)$ in Theorem \ref{thm:decomp}.
Then \eqref{eq:again} again holds, so that for all $m \in \{1, \ldots, M\}$: 
\[ \expect{\lambda_m T(\alpha^*[r]) - \mu_m(\alpha^*[r])} \leq 0 \]
 Substituting the above into the right-hand-side of \eqref{eq:alg1} gives: 
 \[ \Delta[r] \leq B \]
 from which we then obtain rate stability in the same way as in the proof for Algorithm 2. 
 \end{proof} 

\section*{Appendix D --- Proof of the Queue Size Bound} 

Here we show that if $\bv{\lambda} \in \rho\Lambda_{\script{A}}$, where $0 \leq \rho < 1$, then 
both Algorithm 1 and Algorithm 2 yield finite average backlog of size $O(1/(1-\rho))$.

\begin{proof} (Queue Bound for Algorithm 1)
Because $\bv{\lambda} \in \rho \Lambda_{\script{A}}$, we have: 
\[ (\lambda_m/\rho) \in \Lambda_{\script{A}} \]
Thus, from Theorem \ref{thm:decomp} there is a randomized algorithm $\alpha^*[r]$ that makes decisions
independent of queue backlogs to yield the following for all $m \in \{1, \ldots, M\}$: 
\begin{equation} \label{eq:rho-again} 
 \frac{\lambda_m}{\rho} \leq \frac{\expect{\mu_m(\alpha^*[r])}}{\expect{T(\alpha^*[r])}} 
 \end{equation}
 Define $T^* \defequiv \expect{T(\alpha^*[r])}$.  Using this and rearranging the above gives: 
 \[ \expect{\mu_m(\alpha^*[r])} \geq \lambda_mT^*/\rho \: \: \forall m \in \{1, \ldots, M\} \]
 Substituting the above into \eqref{eq:alg1} gives: 
 \[ \Delta[r] \leq B + \sum_{m=1}^MQ_m[r]T^*\lambda_m(1 - 1/\rho) \]
 Taking expectations and using the definition of $\Delta[r]$ gives: 
 \[ \expect{L[r+1]} - \expect{L[r]} \leq B + \sum_{m=1}^M\expect{Q_m[r]}T^*\lambda_m(1-1/\rho) \]
 Summing over $r \in \{0, \ldots, R-1\}$ (for any integer $R>0$) gives: 
 \[ \expect{L[R]} - \expect{L[0]} \leq BR + \sum_{r=0}^{R-1}\sum_{m=1}^M\expect{Q_m[r]}T^*\lambda_m(1-1/\rho) \]
 Using the fact that $\expect{L[r]}\geq0$ and $\expect{L[0]} = 0$, dividing by $R$, and rearranging terms gives:   
 \[ \frac{1}{R}\sum_{r=0}^{R-1} \sum_{m=1}^M\lambda_m\expect{Q_m[r]} \leq \frac{B\rho}{T^*(1-\rho)}\]
 Because $T^*\geq 1$, the above bound can be simplified to $B\rho/(1-\rho)$. 
 The above holds for all $R$, and so the expected queue backlog is $O(1/(1-\rho))$.  Further, from \cite{sno-text} we
 can derive that the following holds with probability 1: 
 \[ \limsup_{R\rightarrow\infty} \frac{1}{R}\sum_{r=0}^{R-1}\sum_{m=1}^M\lambda_mQ_m[r] \leq \frac{B\rho}{T^*(1-\rho)}\]
\end{proof} 

\begin{proof} (Queue Bound for Algorithm 2) 
Recall that \eqref{eq:ratio-compare2} holds for every frame $r$ and all possible $\bv{Q}[r]$ for Algorithm 2. 
Using $\alpha^*[r]$ as an algorithm that makes randomized decisions on frame $r$ that are independent 
of $\bv{Q}[r]$ gives: 
\begin{eqnarray}
\frac{\expect{\sum_{m=1}^MQ_{m}[r]\mu_{m}(\alpha[r])|\bv{Q}[r]}}{\expect{T(\alpha[r])|\bv{Q}[r]}} \geq \nonumber \\
\frac{\sum_{m=1}^MQ_{m}[r]\expect{\mu_{m}(\alpha^*[r])}}{\expect{T(\alpha^*[r])}} \label{eq:ratio-compare3} 
\end{eqnarray}
where we have removed the conditional expectations on the right-hand-side.  Because $(\lambda_m/\rho) \in \Lambda_{\script{A}}$, 
we know that there is an algorithm that makes independent and randomized decisions to yield \eqref{eq:rho-again}. 
Plugging \eqref{eq:rho-again} into the right-hand-side of \eqref{eq:ratio-compare3} gives: 
\begin{eqnarray*}
\frac{\expect{\sum_{m=1}^MQ_{m}[r]\mu_{m}(\alpha[r])|\bv{Q}[r]}}{\expect{T(\alpha[r])|\bv{Q}[r]}} \geq 
\sum_{m=1}^MQ_{m}[r]\lambda_m/\rho
\end{eqnarray*}
Rearranging gives: 
\begin{eqnarray*}
\sum_{m=1}^MQ_m[r]\expect{\mu_m(\alpha[r])|\bv{Q}[r]} \geq \\
 \frac{1}{\rho}\sum_{m=1}^MQ_m[r]\expect{\lambda_mT(\alpha[r])|\bv{Q}[r]}
 \end{eqnarray*}
 Using this in the drift bound of Lemma \ref{lem:causal-drift} gives: 
 \[ \Delta[r] \leq B + \sum_{m=1}^MQ_m[r]\expect{\lambda_m T(\alpha[r]) - \frac{\lambda_m}{\rho}T(\alpha[r])|\bv{Q}[r]} \]
 That is: 
 \begin{eqnarray*}
  \Delta[r] &\leq& B + \sum_{m=1}^MQ_m[r]\lambda_m(1-1/\rho)\expect{T(\alpha[r])|\bv{Q}[r]} \\
  &\leq& 
  B + \sum_{m=1}^MQ_m[r]\lambda_m(1-1/\rho) 
 \end{eqnarray*}
 where we have used the fact that $\expect{T(\alpha[r])|\bv{Q}[r]} \geq 1$, and $1-1/\rho \leq 0$. 
 We thus have by the same argument as in the previous proof that for any $R>0$: 
 \[ \frac{1}{R}\sum_{r=0}^{R-1} \sum_{m=1}^M\lambda_m\expect{Q_m[r]} \leq \frac{\rho B}{1-\rho} \]
 and with probability 1: 
 \[ \limsup_{R\rightarrow\infty} \frac{1}{R}\sum_{r=0}^{R-1}\sum_{m=1}^M\lambda_mQ_m[r] \leq \frac{\rho B}{1-\rho} \]
\end{proof} 

\section*{Appendix E --- The Disjoint Cycle Theorem for General Index Coding} 

The weighted compressed graph $\WCG$ was introduced in Section \ref{section:broadcast} for the context
of broadcast relay networks, being special cases of index coding problems where each packet has exactly one 
incoming link and exactly one outgoing link.   However, such a graph is also useful for general index coding
problems, with general directed bipartite graphs $\script{G}$.

Consider any such general directed bipartite graph $\script{G}$ with user nodes $\script{N}$ and packet nodes
$\script{P}$.  Define $\WCG$ similarly: 
It is a graph on the user nodes $\script{N}$, a link $(i,j)$ exists if and only if user $i$ has a packet as side information that 
is wanted by user $j$, and each link is weighted by $P_{ij}$, the (integer) number of distinct packets of this type.  
An example of a graph $\script{G}$ and its weighted compressed graph $\WCG$ is given in Fig. \ref{fig:compressed-graphs}. 

\begin{figure}[htbp]
   \centering
   \includegraphics[height=1.5in, width=3in]{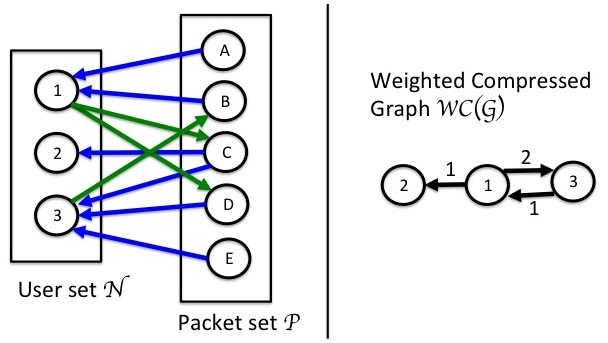} 
   \caption{An example graph $\script{G}$ and its weighted compressed graph  $\WCG$.}
   \label{fig:compressed-graphs}
\end{figure}

Note that weighted compressed graphs in this case typically include less information than
  the original graph $\script{G}$.  For example,  the existence of a link $(a,b)$ with weight $P_{ab}$ 
  in the  graph $\WCG$  tells us that node $a$ has $P_{ab}$ distinct packets that are wanted by node $b$, 
   but does not specify which packets these are, or if these packets also take part in the weight count on other links
   of $\WCG$.  Indeed, in the example of Fig. \ref{fig:compressed-graphs}, packet $C$ is the packet corresponding
   to the weight 1 on link $(1,2)$ in $\WCG$, and packets $C$ and $D$ correspond to the weight 2 on link $(1,3)$
   in $\WCG$, so that the same packet $C$ is included in the weight for two different links.  Note also that, unlike  broadcast
   relay networks, the sum of the weights of $\WCG$ is not necessarily equal to the number of packets $P$ in $\script{G}$.
  
  \begin{lem}  \label{lem:compress} The original demand graph $\script{G}$ is acyclic if and only if $\WCG$ is acyclic.    
  \end{lem} 
  
  \begin{proof} 
 We show that $\script{G}$ has cycles
  if and only if $\WCG$ has cycles.   The proof can be understood directly from Fig. \ref{fig:cngproof}. 
  Suppose that $\WCG$ has a cycle of length $K$, as shown in 
  Fig. \ref{fig:cngproof}a.  Relabeling the nodes, this cycle uses nodes $\{1, \ldots, K\}$, and has 
  links $\{(1,2), (2, 3), \ldots, (K-1, K), (K, 1)\}$.  However, for each link $(a,b)$ of this cycle, there must 
  be a packet $x_{ab} \in \script{P}$ such that $x_{ab} \in \script{H}_a$ (so the graph $\script{G}$ has a directed
  link from user node $a$ to  packet node $x_{ab}$), and $x_{ab} \in \script{R}_b$ (so the graph $\script{G}$ has
  a directed link from packet node $x_{ab}$ to user node $b$).  
  This means that the original graph $\script{G}$ has a cycle, as depicted in Fig. \ref{fig:cngproof}b.
  
  Conversely, if the bipartite graph $\script{G}$ has a cycle, it must alternate between user nodes and 
  packet nodes, with a structure as depicted in Fig. \ref{fig:cngproof}b.  From this, it is clear that $\WCG$ also
  has a cycle.   
    \end{proof} 
    
    \begin{figure}[htbp]
       \centering
       \includegraphics[height=1.5in, width=3.5in]{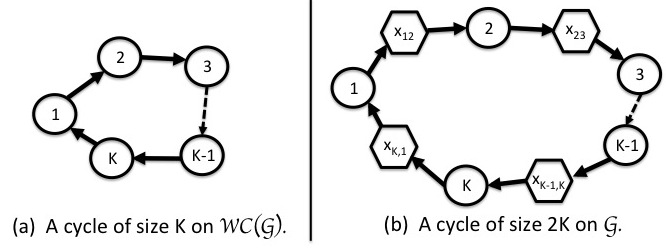} 
       \caption{An illustration for the proof of Lemma \ref{lem:compress}.}
       \label{fig:cngproof}
    \end{figure}
  
  Lemma \ref{lem:compress} 
  makes it easier to see whether or not a given graph $\script{G}$ is acyclic.  That the graph $\script{G}$ in Fig. \ref{fig:compressed-graphs} has cycles is immediately apparent from its much simpler weighted compressed graph $\WCG$, which has a single
  cycle of size $2$ consisting of nodes $1$ and $3$.  As another example, it may not be immediately clear
  that the graph $\script{G}$ 
  in Fig. \ref{fig:bipartite-dag} is acyclic.  However, its graph $\WCG$ 
  has only 3 nodes $\script{N} = \{1, 2, 3\}$ and two links $(1,3)$ and $(3,2)$, and from this the acyclic structure is 
  obvious.    These compressed graphs are also useful for coding,  even when they have cycles, as shown next. 

We say that a packet is a \emph{unicast packet} if it has at most one outgoing link (so that it is intended for only
one destination user), and a packet is a \emph{multicast packet} if it is intended for more than one destination user. 
We say that the weighted graph $\WCG$ has disjoint cycles if no link participates in more than one cycle. 
Suppose now that $\WCG$ has disjoint cycles, and let $K$ be the number
of such cycles.  As before, for each $k \in \{1, \ldots, K\}$, we let $\script{C}^{(k)}$ represent the set of links for the $k$th 
cycle, and let $w_{min}^{(k)}$ represent the weight of the minimum weight link in $\script{C}^{(k)}$. Note
that removing this link from each cycle results in a new graph that is acyclic.  This is used in the following theorem. 

\begin{thm} \label{thm:disjoint} Let $\script{G}$ be a demand graph with $N$ nodes and $P$ packets.  Suppose that 
$\WCG$ has disjoint
cycles, that there are $K$ such cycles, and that  all packets of these cycles are distinct and are unicast packets. Then: 
\[ T_{min}(\script{G}) = P-\sum_{k=1}^K w_{min}^{(k)} \]
where $P$ is the number of packets in $\script{G}$.
  Furthermore, 
the minimum clearance time can be achieved by performing cyclic coding $w_{min}^{(k)}$ times for each cycle
$k \in \{1, \ldots, K\}$, and then transmitting all the remaining packets without coding. 
\end{thm} 
\begin{proof} 
For each cycle $k \in \{1, \ldots, K\}$, select a link with a link weight equal to the minimum
link weight $w_{min}^{(k)}$ (breaking ties arbitrarily).  Let $\script{P}^{(k)}$ be the set of all packets associated with this link. 
All packets in  $\cup_{k=1}^K \script{P}^{(k)}$ are distinct (by assumption), and
the total number of these packets is: 
\[ |\cup_{k=1}^K\script{P}^{(k)}| = \sum_{k=1}^Kw_{min}^{(k)}  \]
Now consider the subgraph $\script{G}'$ formed from $\script{G}$ by removing all packet nodes
in the set $\cup_{k=1}^K \script{P}^{(k)}$.   The number of packets $P'$ in this graph is thus: 
\[ P' = P - \sum_{k=1}^Kw_{min}^{(k)} \]
Further, we have $T_{min}(\script{G}') \leq T_{min}(\script{G})$.  
However, note that $\script{WC}(\script{G}')$ is the same as $\WCG$, with the exception that the
min-weight link on each of the $K$ cycles has been removed.  Thus,  $\script{WC}(\script{G}')$ is
acyclic, so that $\script{G}'$ is acyclic, and by Theorem \ref{thm:acyclic} we have $T_{min}(\script{G}') = P'$. 
It follows that: 
\[ T_{min}(\script{G}) \geq P'  \]

Thus, any algorithm for clearing all packets in the graph $\script{G}$ must use at least $P'$ slots.  
We now show that $T_{min}(\script{G}) = P'$ 
by designing a simple cyclic coding scheme to clear all packets in the original graph $\script{G}$ in exactly 
$P'$ slots.  By assumption, all packets that participate in cycles are distinct unicast packets, 
and hence they only need to be delivered to one destination, as defined by the link of the cycle they are in.  
This includes all packets in the set $\cup_{k=1}^K\script{P}^{(k)}$. 
We now use the obvious strategy:   
For each cycle $k \in \{1, \ldots, K\}$, we choose an undelivered packet
$p \in \script{P}^{(k)}$, and 
use a cyclic coding operation that clears $S^{(k)}$ distinct packets (including packet $p$) in
 $S^{(k)}-1$ slots, where $S^{(k)}$ is the size of cycle $k$. This can be done because packet $p$ is from the link 
of cycle $k$ with the fewest packets, and so there are always remaining packets on the other links of the cycle to 
use in the coding.  Once this is done for all cycles and for all packets on the min-weight link of each cycle, we then
transmit the remaining packets uncoded.  The total number of transmissions is thus equal to $P$ minus the savings
of $\sum_{k=1}^Kw_{min}^{(k)}$ from all of the cyclic coding operations.  Thus, this takes exactly $P - \sum_{k=1}^Kw_{min}^{(k)} = P'$ slots. 
\end{proof}

 \begin{cor} 
 If the demand graph $\script{G}$ has $P$ packets, but only two users (so that $\script{N} = \{1, 2\}$), then:  
 \[ T_{min}(\script{G}) = P - w_{min} \]
 where $w_{min} \defequiv \min[P_{12}, P_{21}]$, being the weight of the 
 min-weight
 link in the graph $\WCG$.
 Further, this minimum clearance time can be achieved by performing cyclic coding over the packets associated with this 
 min weight link, and then transmitting (uncoded) all remaining packets.  
 \end{cor} 
 \begin{proof} 
 It suffices to show that $\WCG$ has disjoint cycles, and that all packets that are part of these cycles are distinct unicast packets. 
 The graph $\WCG$ has only 2 nodes and hence at most one cycle, so it clearly satisfies the 
 disjoint cycle criterion.  All packets of the cycle are clearly distinct. Indeed, two packets on different links of the 
 cycle must have one packet on link $(1,2)$ and the other on link $(2,1)$, and hence the first is desired by user 2, 
 while the second is \emph{not} desired by user 2 (so they cannot be the same packet). 
 It remains to show that all packets in the cycle are unicast packets. If there are no cycles, we are
 done.  If there is a cycle, this involves link $(1,2)$ and link $(2,1)$.  Any packet associated with the link $(1,2)$ 
 must already be in the set of packets that node $1$ has, and so this packet only requires transmission to node $2$ (not to node $1$).  Thus, this packet must be a unicast packet.  Similarly, any packet associated with link $(2,1)$ must be  a unicast packet. 
 \end{proof} 

\section*{Appendix F --- Optimality of Cyclic Coding for Relay Networks with $N=3$}

Consider now the special case of broadcast relay networks, where each packet node of  
the graph $\script{G}$ has exactly one incoming link and one outgoing link.  There are $N=3$ users
and $P$ packets (where $P$ is a positive integer). This can be exactly 
represented by the weighted
compressed graph $\WCG$ with link weights $P_{ij}$, where $P = \sum_{i=1}^N\sum_{j=1}^NP_{ij}$.  We want to show that minimum
clearance time can be achieved by cyclic coding (using either direct transmission, 2-cycle code actions, 
or 3-cycle code actions).  If the 3-node graph $\WCG$ has
disjoint cycles, we are done (recall Theorem \ref{thm:disjoint-cycles}). 

\begin{figure}[htbp]
   \centering
   \includegraphics[height=2.5in, width=3in]{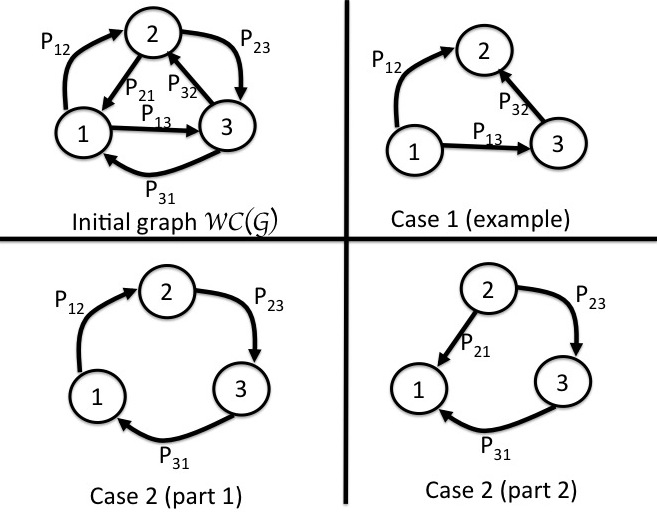} 
   \caption{An illustration of the general broadcast relay graph $\WCG$ with $N=3$ users, and the two cases required
   for the proof.}
   \label{fig:3case}
\end{figure}

Consider now the general case with possibly non-disjoint cycles, as shown in Fig. \ref{fig:3case}. 
To represent the general case, we allow link weights $P_{ab}$ to possibly be $0$ (a weight of $0$ is equivalent
to the absence of a link). 
First define $min_{12}$, $min_{23}$, $min_{31}$ as the 
weight of the min-weight link for each of the three possible 2-cycles: 
\begin{eqnarray*}
min_{12} &=& \min[P_{12}, P_{21}] \\
min_{23} &=& \min[P_{23}, P_{32}] \\
min_{31} &=& \min[P_{31}, P_{13}] 
\end{eqnarray*}
Now prune the graph $\WCG$ by removing
the min-weight link for each 2-cycle.  This results in a graph $\script{WC}(\script{G}')$ with 3 nodes and (at most)
3 links, and with a total number of 
packets equal to $P- min_{12} - min_{23} - min_{31}$, as shown in Cases 1 and 2 in the figure.  
We have two cases. 

Case 1: The resulting graph $\script{WC}(\script{G}')$  is acyclic.  An example of this case is shown as case 1 in Fig. \ref{fig:3case}. 
Thus, we know $T_{min}(\script{G}) \geq T_{min}(\script{G}') = P -  min_{12} - min_{23} - min_{31}$. 
However, it is easy to see this clearance time bound 
can be achieved by using 2-cycle code actions on each of the three 2-cycles, and then
transmitting the remaining packets uncoded. 

Case 2: The resulting graph $\script{WC}(\script{G}')$ consists of a single  3-cycle.  The cycle must either be clockwise
or counter-clockwise.  Without loss of generality, assume clockwise (see Fig. \ref{fig:3case}, case 2 part 1).  
Note that: 
\begin{eqnarray}
P_{12} \geq P_{21} \: , \:  P_{23} \geq P_{32} \: , \: P_{31} \geq P_{13} \label{eq:minweight-links} 
\end{eqnarray}
This is because we have formed $\script{WC}(\script{G}')$ by removing the min-weight link on each of the three 2-cycles. 

Let $z = \min[P_{12} - P_{21}, P_{23} - P_{32}, P_{31} - P_{13}]$, so that $z$ is a non-negative integer. 
Without loss of generality, assume the min value for $z$ is achieved by 
link $(1,2)$, so that $z = P_{12} - P_{21}$. 
To $\script{WC}(\script{G}')$, add back the link $(2,1)$ (with weight $P_{21}$), 
and remove the link $(1,2)$, to yield a graph $\script{WC}(\script{G}'')$ that is an acyclic subgraph of the 
original graph $\WCG$, 
as shown in case 2 part 2 of Fig. \ref{fig:3case}. The acyclic subgraph $\script{WC}(\script{G}'')$ 
contains exactly $P_{21} + P_{23} + P_{31}$ packets. Thus, the 
minimum clearance time of the original graph $\WCG$ 
is at least $P_{21} + P_{23} + P_{31}$.  However, this can easily be achieved.   
Do the following:  Perform 2-cycle code actions on each of the three 2-cycles of the original graph $\WCG$, 
to remove a number of 
packets on each cycle equal to the min weight link of that cycle. This removes 
$2P_{21}  + 2P_{32}  + 2P_{13}$
packets in $P_{21} + P_{32} + P_{13}$ slots (recall the three min weights are given by 
\eqref{eq:minweight-links}). 
Then perform the 3-cycle code action to remove $3z$ packets
in $2z$ slots. Then perform direct transmission to remove the remaining packets, being
a total of $x = P  -2P_{21} - 2P_{32} - 2P_{13}  - 3z$.  The total number of transmissions is: 
\begin{eqnarray*}
&&\hspace{-.4in}P_{21} + P_{32} + P_{13}  + 2z + x  \\
&=&  P_{21} + P_{32} + P_{13}  
+ 2z \\
&&+ P  -2P_{21} - 2P_{32} - 2P_{13}  - 3z \\
&=& P - P_{21} - P_{32} - P_{13} - z \\
&=& P_{12} + P_{23} + P_{31} - z \\
&=& P_{12} + P_{23} + P_{31} - (P_{12} - P_{21}) \\
&=& P_{21} + P_{23} + P_{31} 
\end{eqnarray*}
and so the above scheme is optimal. 

\bibliographystyle{unsrt}
\bibliography{../../latex-mit/bibliography/refs}
\end{document}